\documentclass [11pt]{article}
\usepackage[applemac]{inputenc}
\usepackage{cite}
\usepackage{amssymb}
\usepackage{amsmath}
\usepackage{graphicx}
\usepackage{fullpage}
\usepackage{relate}
\newtheorem{lemma}{Lemma}
\newtheorem{theorem}{Theorem}

\newcommand{\qed}{\hfill\ensuremath{\Box}\medskip\\\noindent}
\newenvironment{proof}{\noindent\emph{Proof. }}{\qed}

\newcommand{\leaf}{\ensuremath{\mathrm{leaf}}}

\newcommand{\rank}{\ensuremath{\mathrm{rank}}}

\newcommand{\lab}{\ensuremath{\mathrm{label}}}
\newcommand{\nca}{\ensuremath{\mathrm{nca}}}
\newcommand{\lnca}{\ensuremath{\mathrm{lnca}}}

\newcommand{\lcp}{\ensuremath{\mathrm{lcp}}}

\newcommand{\strdepth}{\ensuremath{\mathrm{strdepth}}}
\newcommand{\LCP}{\ensuremath{\mathrm{LCP}}}
\newcommand{\LCPP}{\ensuremath{\mathrm{LCPP}}}
\newcommand{\LDP}{\ensuremath{\mathrm{LDP}}}

\newcommand{\bin}{\ensuremath{\mathrm{bin}}}

\newcommand{\rmb}{\ensuremath{\mathrm{rmb}}}
\newcommand{\lmb}{\ensuremath{\mathrm{lmb}}}
\newcommand{\lsmear}{\ensuremath{\mathrm{lsmear}}}
\newcommand{\rsmear}{\ensuremath{\mathrm{rsmear}}}

\newcommand{\lcpp}{\ensuremath{\mathrm{lcpp}}}
\newcommand{\ldp}{\ensuremath{\mathrm{ldp}}}

\newcommand{\dom}{\ensuremath{\mathrm{dom}}}

\newcommand{\lex}{\ensuremath{\mathrm{lex}}}

\newcommand{\Zip}{\ensuremath{\textsc{Zip}}}
\newcommand{\Unzip}{\ensuremath{\textsc{Unzip}}}

\newcommand{\Lnca}{\ensuremath{\textsc{Lnca}}}

\newcommand{\Merge}{\ensuremath{\textsc{Merge}}}
\newcommand{\Sort}{\ensuremath{\textsc{Sort}}}
\newcommand{\SortedMap}{\ensuremath{\textsc{$\Map^{\triangle}$}}}
\newcommand{\Map}{\ensuremath{\textsc{Map}}}

\newcommand{\ceil}[1]{\left\lceil{#1}\right\rceil}
\newcommand{\floor}[1]{\left\lfloor{#1}\right\rfloor}
\renewcommand{\angle}[1]{\langle{#1}\rangle}

\title{Faster Approximate String Matching for Short Patterns}

\author{Philip Bille\thanks{Technical University of Denmark, DK-2800 Kgs. Lyngby, Denmark, Email: {\tt phbi@imm.dtu.dk}. Supported by the Danish Agency for Science, Technology, and Innovation}}
\date{\today}

\begin{document}
\maketitle

\begin{abstract}
  We study the classical approximate string matching problem, that is,
  given strings $P$ and $Q$ and an error threshold $k$, find all
  ending positions of substrings of $Q$ whose edit distance to $P$ is
  at most $k$. Let $P$ and $Q$ have lengths $m$ and $n$,
  respectively. On a standard unit-cost word RAM with word size $w
  \geq \log n$ we present an algorithm using time
$$
O\left(nk \cdot \min\left(\frac{\log^2 m}{\log n},\frac{\log^2 m\log w}{w}\right) + n\right)
$$
When $P$ is short, namely, $m = 2^{o(\sqrt{\log n})}$ or $m =
2^{o(\sqrt{w/\log w})}$ this improves the previously best known time
bounds for the problem. The result is achieved using a novel
implementation of the Landau-Vishkin algorithm based on tabulation and
word-level parallelism.
\end{abstract}

\section{Introduction}
Given strings $P$ and $Q$ and an \emph{error threshold} $k$, the
\emph{approximate string matching problem} is to report all ending
positions of substrings of $Q$ whose \emph{edit distance} to $P$ is at
most $k$. The edit distance between two strings is the minimum number
of insertions, deletions, and substitutions needed to convert one
string to the other. Approximate string matching is a classical and
well-studied problem in combinatorial pattern matching with a wide
range of applications in areas such as bioinfomatics, network
traffic analysis, and information retrieval.

Let $m$ and $n$ be the lengths of $P$ and $Q$, respectively, and
assume without loss of generality that $k < m \leq n$. The classic
textbook solution to the problem, due to Sellers~\cite{Sellers1980},
fills in an $(m+1) \times (n+1)$ \emph{distance matrix} $C$ such that
$C_{i,j}$ is the smallest edit distance between the $i$th prefix of
$P$ and a substring of $Q$ ending at position $j$. Using dynamic
programming, we can compute each entry in $C$ in constant
time leading to an algorithm using $O(nm)$ time.

Several improvements of this algorithm are known. Masek and
Paterson~\cite{MP1980} showed how to compactly encode and tabulate
solutions to small submatrices of the distance matrix. We can then
traverse multiple entries in the table in constant time leading to an
algorithm using $O(nm/\log^2 n + n)$ time. This bound assumes constant
size alphabets. For general alphabets, the best bound is $O(nm(\log
\log n)^2/\log^2 n + n)$~\cite{BFC2008}. This tabulation technique is
often referred to as the \emph{Four Russian technique} after Arlazarov
et al.~\cite{ADKF1970} who introduced it for boolean matrix
multiplication. Alternatively, several algorithms using the arithmetic
and logical operations of the word RAM to simulate the dynamic program
have been suggested~\cite{BYG1992, WM1992b, Wright1994, BYN1996,
  Myers1999, HN2005}. This technique is often referred to as
\emph{word-level parallelism} or \emph{bit-parallelism}. The best
known bound is due to Myers~\cite{Myers1999} who gave an algorithm
using $O(nm/w + n)$ time.  In terms of $n$ and $m$ alone, these are
the best known bounds for approximate string matching. However, if we
take into account the error threshold $k$, several faster algorithms
are known~\cite{UW1993, Ukkonen1985b, Myers1986, GG88, LV1989, GP90,
  SV96, CH2002}. These algorithms exploit properties of the diagonals
of the distance matrix $C$ and are therefore often called
\emph{diagonal transition algorithms}. The best known bound is due to
Landau and Vishkin~\cite{LV1989} who gave an $O(nk)$
algorithm. Compared to the algorithms by Masek and Paterson and by
Myers, the Landau-Vishkin algorithm (abbreviated LV-algorithm) is
faster for most values of $k$, namely, whenever $k = o(m/\log^{2} n)$
or $k = o(m/w)$. For $k = O(m^{1/4})$, Cole and Hariharan showed that
it is even possible to solve approximate string matching in $O(n)$
time. Their algorithm ``filters'' out all but a small set of positions
in $Q$ which are then checked using the LV-algorithm.

All of the above bounds are valid on a unit-cost RAM with $w$-bit
words and a standard instruction set including arithmetic operations,
bitwise boolean operations, and shifts. Each word is capable of
holding a character of $Q$ and hence $w \geq \log n$. The space
complexity is the number of words used by the algorithm, not counting
the input which is assumed to be read-only.  For simplicity, we assume
that suffix trees can be constructed in linear time which is true for
any polynomially sized alphabet~\cite{CFM2000}. This assumption is
also needed to achieve the $O(nk)$ bound of the Landau-Vishkin
algorithm~\cite{LV1989}. Without it, additional time for sorting the
alphabet is required~\cite{CFM2000}. All the results presented here
assume the same model.

\subsection{Results}
We present a new algorithm for approximate string matching achieving
the following bounds.
\begin{theorem}\label{thm:main}
Approximate string matching for strings $P$ and $Q$ of length $m$ and $n$, respectively, with error threshold $k$ can be solved 
\begin{itemize}
  \item[(i)] in time $O(nk\cdot\frac{\log^2 m}{\log n} + n)$ and space $O(n^\epsilon + m)$, for any constant $\epsilon > 0$, and 
  \item[(ii)] in time $O(nk\cdot\frac{\log^2 m\log w}{w} + n)$ and space $O(m)$.
\end{itemize}
\end{theorem}
When $P$ is short, namely, $m = 2^{o(\sqrt{\log n})}$ or $m =
2^{o(\sqrt{w/\log w})}$, this improves the $O(nk)$ time bound and
places a new upper bound on approximate string matching. For many
practically relevant combinations of $n$, $m$ and $k$ this
significantly improves the previous results. For instance, when $m$ is
polylogarithmic in $n$, that is, $m = O(\log^c n)$ for a constant $c >
0$, Theorem~\ref{thm:main}(i) gives us an algorithm using time $O(nk
\cdot \frac{(\log \log n)^{2}}{\log n} + n)$. This is almost a
logarithmic speed-up of $O(\frac{\log n}{(\log \log n)^2})$ over the
$O(nk)$ bound. Note that the exponent $c$ only affects the constants
in asymptotic time bound. For larger $m$, the speed-up smoothly
decreases until $m = 2^{\Theta(\sqrt{\log n})}$, where we arrive at
the $O(nk)$ bound.

The algorithm for Theorem~\ref{thm:main}(i) tabulates certain
functions on $\epsilon \log n$ bits which lead to the additional
$O(2^{\epsilon\log n}) = O(n^{\epsilon})$ space. The algorithm for
Theorem~\ref{thm:main}(ii) instead uses word-level parallelism and
therefore avoids the additional space for lookup tables. Furthermore,
for $w = O(\log n)$, Theorem~\ref{thm:main}(ii) gives us an algorithm
using time $O(nk \cdot \frac{\log^{2} m\log \log n}{\log n} +
n)$. This is a factor $O(\log \log n)$ slower than
Theorem~\ref{thm:main}(i). However, the bound increases with $w$ and
whenever $w\log w = \omega(\log n)$, Theorem~\ref{thm:main}(i) is the
best time bound.

\subsection{Techniques}
The key idea to obtain our bounds is a novel implementation of the
LV-algorithm that reduces approximate string matching to $2$
operations on a compact encoding of the ``state'' of the LV-algorithm.
We show how to implement these operations using tabulation for
Theorem~\ref{thm:main}(i) or word-level parallelism for
Theorem~\ref{thm:main}(ii). As discussed above, several improvements
of Sellers classical dynamic programming algorithm~\cite{Sellers1980}
based on tabulation and word-level parallelism are known. However, for
diagonal transition algorithms no similar tabulation or word-level
parallelism improvements exists. Achieving such a result is also
mentioned as an open problem in a recent survey by
Navarro~\cite[p.61]{Navarro2001a}. The main problem is the complicated
dependencies in the computation of the LV-algorithm. In particular, in
each step of the LV-algorithm we map entries in the distance matrix to
nodes in the suffix tree, answer a nearest common ancestor query, and
map information associated with the resulting node back to an entry in
the distance matrix. To efficiently compute this information in
parallel, we introduce several new techniques. These techniques differ
significantly from the techniques used to speed-up Sellers algorithm, 
and we believe that some of them might be of independent interest. For
example, we give a new algorithm to efficiently evaluate a compact
representation of a function on several inputs in parallel. We also
show how to use a recent distributed nearest common ancestor data
structure to efficiently answer multiple nearest common ancestor
queries in parallel.

The results presented in this paper are mainly of theoretical
interest. However, we believe that some of the ideas have practical
relevance. For instance, it is often reported that the nearest common
ancestor computations make the LV-algorithm unsuited for practical
purposes~\cite{Navarro2001a}. With our new algorithm, we can compute
several of these in parallel and thus target this bottleneck.

\subsection{Outline}
The paper is organized as follows. In Section~\ref{sec:preliminaries}
we review the basic concepts and the LV-algorithm. In
Section~\ref{sec:bittricks} we introduce the packed representation and
the key operations needed to manipulate it. In
Section~\ref{sec:reduction} we reduce approximate string matching to
two operations on the packed representation. Finally, in
Sections~\ref{sec:lncamap} and Section~\ref{sec:word} we present our
tabulation based algorithm and word-level parallel algorithm for these
operations.

\section{Preliminaries}\label{sec:preliminaries}
We review the necessary concepts and the basic algorithms for
approximate string matching. We will use these as a starting point for
our own algorithms.

\subsection{Strings, Trees, and Suffix Trees}
Let $S$ be a string of length $|S|$ on an alphabet $\Sigma$. The
character at position $i$ in $S$ is denoted by $S[i]$, and the
substring from position $i$ to $j$ is denoted by $S[i,j]$. The
substrings $S[1,j]$ and $S[i, |S|]$ are the \emph{prefixes} and
\emph{suffixes} of $S$, respectively. The \emph{longest common prefix}
of two strings is the common prefix of maximum length.

Let $T$ be a rooted tree with $|T|$ nodes. A node $v$ in $T$ is an
\emph{ancestor} of a node $w$ if $v$ is on the path from the root to
$w$ (including $v$ itself). A node $z$ is a \emph{common ancestor} of
nodes $v$ and $w$ if $z$ is an ancestor of both. The \emph{nearest
  common ancestor} of $v$ and $w$, denoted $\nca(v,w)$, is the common
ancestor of $v$ and $w$ of maximum depth in $T$. With linear space
and preprocessing time, we can answer $\nca$ queries in constant
time~\cite{HT1984} (see also~\cite{BFC2000,AGKR2004}).

The \emph{suffix tree} for $S$, denoted $T_S$, is the compacted trie
storing all suffixes of $S$~\cite{Gusfield1997}. Each edge $e$ in
$T_S$ is associated with a substring of $S$, called the
\emph{edge-label} of $e$.  The concatenation of edge-labels on a path
from the root to a node $v$ is called the \emph{path-label} of $v$.
The \emph{string-depth} of $v$, denoted $\strdepth(v)$, is the length
of the path-label of $v$. The $i$th suffix of $S$ is represented by
the unique leaf in $T_S$ whose path-label is $S[i, |S|]$, and we denote
this leaf by $\leaf(i)$. The suffix tree uses linear space and can be
constructed in linear time for polynomially sized
alphabets~\cite{CFM2000}.

A useful property of suffix trees is that for any two leaves
$\leaf(i)$ and $\leaf(j)$, the path label of the node $\nca(\leaf(i),
\leaf(j))$ is longest common prefix of the suffixes $S[i,|S|]$ and
$S[j,|S|]$~\cite{Gusfield1997}. Hence, if we construct a nearest
common ancestor data structure for $T_S$ and compute the string depth
for each node in $T_S$, we can compute the length of the longest common
prefix of any two suffixes in constant time.

For a set of strings $S_1, \ldots, S_l$ it is straightforward to
construct a suffix tree $T_{S_1, \ldots, S_l}$ storing all suffixes of
each string in $S_1, \ldots, S_l$~\cite{Gusfield1997}. A suffix tree
of more than one string is often called a \emph{generalized suffix
  tree}~\cite{Gusfield1997}. The space for $T_{S_1, \ldots, S_l}$ is
linear in the total length of the strings.

\subsection{Algorithms for Approximate String Matching}
Recall that $|P| = m$ and $|Q| = n$ and $k$ is the error
threshold. The algorithm by Sellers~\cite{Sellers1980} fills in a
$(m+1) \times (n+1)$ matrix $C$ according to the following rules:
\begin{equation}\label{Sellersrecursion}
\begin{aligned}
  C_{i,0}  &= i  \qquad 0 \leq i \leq m \\ 
  C_{0,j}  &= 0   \qquad 0 \leq j \leq n \\
  C_{i,j}   &= \min(C_{i-1, j-1} + \delta(p_i, t_j), C_{i-1, j} + 1,
  C_{i, j-1} + 1) \qquad 1 \leq i \leq m, 1 \leq j \leq n
\end{aligned}
\end{equation}
For any pair of characters $a$ and $b$, $\delta(a, b) = 0$ if $a = b$
and $1$ otherwise. An example of a matrix is shown in
Figure~\ref{fig:dynmatrix}. Note that the above rules are the same as
for the classical dynamic program for the well-known edit distance
problem~\cite{WF1974}, except for the boundary condition $C_{0,j} =
0$. The entry $C_{i,j}$ is the minimum edit distance between $P[1,i]$
and any substring of $Q$ ending at position $j$. Hence, there is a match of $P$
with a most $k$ edits that ends at $Q[j]$ iff $C_{m, j} \leq k$. Using
dynamic programming, we can compute each entry in constant time
leading to an $O(nm)$ solution.

Landau and Vishkin~\cite{LV1989} presented a faster algorithm to
compute essentially the same information as
in~\eqref{Sellersrecursion}. We will refer to this algorithm as the
LV-algorithm in the rest of the paper. Define the \emph{diagonal} $d$
of $C$ to be the set of entries $C_{i,j}$ such that $j-i = d$. Given a
diagonal $d$ and integer $e$, define the diagonal position $L_{d,e}$
to be the maximum $i$ such that $C_{i,j} = e$ and $C_{i,j}$ is on
diagonal $d$.  There is a match of $P$ with a most $k$ edits that ends
at $Q[d+m]$ iff $L_{d,e} = m$, for some $e\leq k$. Let $\lcp(i,j)$
denote the length of the longest common prefix of $P[i,m]$ and
$Q[j,n]$.  Using the clever observation that entries in a diagonal are
non-decreasing in the downwards direction, Landau and Vishkin gave the
following rules to compute $L_{d,e}$.
\begin{subequations}\label{LVrecursion}
  \begin{align}
  L_{d,-1}  &= L_{n+1, e} = -1 \qquad \text{for $e \in \{-1, \ldots, k\}$ and $d \in \{0, \ldots, n\}$}   \label{LVrecursion1}\\ 
  L_{d,|d| - 2}  &= |d| - 2 \qquad \text{for $d \in \{-(k+1), \ldots, -1\}$}   \label{LVrecursion2}\\
  L_{d,|d| - 1}  &= |d| - 1 \qquad \text{for $d \in \{-(k+1), \ldots, -1\}$}   \label{LVrecursion3}\\
  L_{d,e} &= z + \lcp(z + 1, d + z + 1) \label{LVrecursion4}\\ 
  \text{where } z &= \min(m, \max(L_{d, e-1} + 1, L_{d-1, e-1},
  L_{d+1, e-1} + 1)) \label{LVrecursion5}
\end{align}
\end{subequations}
Lines \eqref{LVrecursion1}, \eqref{LVrecursion2}, and
\eqref{LVrecursion3} are boundary conditions. Lines
\eqref{LVrecursion4} and \eqref{LVrecursion5} determine $L_{d,e}$ from
$L_{d, e-1}$, $L_{d-1, e-1}$, $L_{d+1, e-1}$, and the length of the
longest common prefix of $P[z+1, m]$ and $Q[d+z +1,n]$. Hence, we can
compute a matrix $L$ of diagonal positions by iteratively computing
the sets of diagonal positions $L_{-1}, L_{0}, \ldots, L_{k}$, where
$L_{e}$ denotes the set of entries in $L$ with error $e$. Since we can
compute $\lcp$ queries in constant time using a nearest common
ancestor data structure, the total time to fill in the $O(nk)$ entries of
$L$ is $O(nk)$. Each set of diagonal positions and the suffix tree
require $O(n)$ space. However, we can always divide $Q$ into
overlapping substrings of length $2m - 2k$ with adjacent substrings
overlapping in $m+k-1$ characters. A substring matching $P$ with at
most $k$ errors must have a length in the range $[m-k, m+k]$ and
therefore all matches are completely contained within a substring.
Applying the LV-algorithm to each of the substrings independently
solves approximate string matching in time $O(n/m \cdot mk) = O(nk)$
as before, however, now the space is only $O(m)$.

\begin{figure}
\centering
\begin{tabular}{| c | c | c | c | c | c | c | c | c |}
  \hline			
    &    & \texttt{s} & \texttt{u} &\texttt{r}  & \texttt{g} & \texttt{e} & \texttt{r}  & \texttt{y} \\\hline
    & 0 & 0 & 0 & 0 & 0 & 0 & 0 & 0 \\\hline
\texttt{s}  & 1 & 0 & 1 & 1 & 1 & 1 & 1 & 1 \\\hline
\texttt{u}  & 2 & 1 & 0 & 1 & 2 & 2 & 2 & 2 \\\hline
\texttt{r}   & 3 & 2 & 1 & 0 & 1 & 2 & 2 & 3 \\\hline
\texttt{v}   & 4 & 3 & 2 & 1 & 1 & 2 & 3 & 3 \\\hline
\texttt{e}   & 5 & 4 & 3 & 2 & 2 & 1 & 2 & 3 \\\hline
\texttt{y}   & 6 & 5 & 4 & 3 & 3 & 2 & 2 & 2 \\\hline  
\end{tabular}

\caption{The dynamic programming matrix $C$ for $P = \texttt{survey}$
  and $Q = \texttt{surgery}$ (adapted from
  Navarro~\cite{Navarro2001a}). $P$ matches $Q$ with edit distance $2$
  at positions $5$, $6$, and $7$. In diagonal $1$, the maximum rows
  indices containing $0$, $1$, and $2$ are $0$, $3$, and $6$,
  respectively. Hence, $L_{1,0} = 0$, $L_{1,1} = 3$, and $L_{1,2} =
  6$.}
\label{fig:dynmatrix}
\end{figure}

\section{Manipulating Bits}\label{sec:bittricks}
In this section we introduce the necessary notation and key primitives for manipulating bit strings. 

Let $x = b_f \ldots b_1$ be a bit string consisting of bits
$b_{1},\ldots, b_{f}$ numbered from right-to-left. The \emph{length}
of $x$, denoted $|x|$, is $f$. We use exponentiation for bit
repetition, i.e., $0^31 = 0001$ and $\cdot$ for concatenation, i.e.,
$001 \cdot 100 = 001100$. In addition to the arithmetic operators $+$,
$-$, and $\times$ we have the operators $\:\&\:$, $\mid$, and $\oplus$
denoting bit-wise `and', `or', and `exclusive-or',
respectively. Moreover, $\overline{x}$ is the bit-wise `not' of $x$
and $x \ll j$ and $x \gg j$ denote standard left and right shift by
$j$ positions. The word RAM supports all of these above operators for
bit strings stored in single words in unit
time~\cite{Hagerup1998}. Note that for bit strings of length $O(w)$
(recall that $w$ is the number of bits in a word) we can still
simulate these instructions in constant time.

We will use the following nearest common ancestor data structure based
on bit string labels in our algorithms.
\begin{theorem}[Alstrup et al.~\cite{AGKR2004}]\label{thm:nca}
  There is a linear time algorithm that labels the $t$ nodes of a tree
  $T$ with bit strings of length $O(\log t)$ bits such that from the
  labels of nodes $v$ and $w$ in $T$ alone, one can compute the label
  of $\nca(v, w)$ in constant time.
\end{theorem}
For our purposes, we will slightly modify the above labeling scheme
such that all labels have the same length $f = O(\log t)$. This is
straightforward to do and we will present one way to do it later in
Section~\ref{sec:labelingscheme}. Let $\lab(v)$ denote the label of a
node $v$ in $T$. The \emph{label nearest common ancestor}, denoted
$\lnca$, is the function given by $\lnca(\lab(v), \lab(w)) =
\lab(\nca(v,w))$ for any pair of labels $\lab(v)$ and $\lab(w)$ of
nodes $v$ and $w$ in $T$. Thus, $\lnca$ maps two bit strings of length
$f$ to a single bit string of length $f$.

\subsection{Packed Sequences}
We often interpret bit strings as sequences of smaller bit strings and
integers. For a sequence $x_1, \ldots, x_r$ of bit strings of length
$f$, define the \emph{$f$-packed sequence} $X = \angle{x_1, \ldots,
  x_r}$ to be the bit string
\begin{equation*}
0 \cdot x_r \cdot 0 \cdot x_{r-1} \cdots 0 \cdot x_2 \cdot 0 \cdot x_1 
\end{equation*}
Each substring $0 \cdot x_i$, $1 \leq i \leq r$, is a
\emph{field}. The leftmost bit of a field is the \emph{test bit} and
the remaining $f$ bits, denoted $X\angle{i} = x_i$, is the
\emph{entry}. The \emph{length} of a $f$-packed sequence is the number
of fields in it. Note that a $f$-packed sequence of length $r$ is
represented by a bit string of length $r(f+1)$. If $x_1, \ldots, x_r$
is a sequence of $f$-bit integers, $\angle{x_1, \ldots, x_r}$ is
interpreted as $\angle{\bin(x_1), \ldots, \bin(x_r)}$, where $\bin(x)$
is the binary encoding of $x$. We represent packed sequences compactly
in words by storing $s = \floor{w/(f+1)}$ fields per word. For our
purposes, we will always assume that ﬁelds are capable of storing the
total number of ﬁelds in the packed sequence, that is, $f \geq \log
r$.  Given another $f$-packed sequence $Y = \angle{y_1, \ldots, y_r}$,
the \emph{zip} of $X$ and $Y$, denoted $X \ddagger Y$, is the
$2f$-packed sequence $\angle{(x_1, y_1), \ldots, (x_r, y_r)}$ (the
tuple notation $(x_{i}, y_{i})$ denotes the bit string $x_{i} \cdot
y_{i}$). Packed sequence representations are well-known within sorting
and data structures (see, e.g., the survey by
Hagerup~\cite{Hagerup1998}). In the following we review some basic
operations on them.

Let $X = \angle{x_{1}, \ldots, x_{s}}$ and $Y = \angle{y_{1}, \ldots,
  y_{s}}$ be $f$-packed sequences of length $s =
\floor{w/(f+1)}$. Hence, $X$ and $Y$ can each be stored in a single word
of $w$ bits. We consider the general case of longer packed sequences
later. Some of our operations require precomputed constants depending
on $s$ and $f$, which we assume are available (e.g., computed at
``compile-time''). If this is not the case, we can always precompute
these constants in time $\log^{O(1)} w$ which is
neglible. 

Elementwise arithmetic operations (modulo $2^f$) and bit-wise
operations are straightforward to implement in $O(1)$ time using the
built-in operations. For example, to compute $\angle{x_1 + y_1 \mod
  2^f, \ldots, x_s + y_s \mod 2^f}$, we add $X$ and $Y$ and clear the
test bits by $\:\&\:$'ing with the constant $I_{s,f} = (10^f)^{s}$
($I_{s,f}$ consists of $1$'s at all test bit positions). The test bit
positions ensures that no overflow bits from the addition can affect
neighbouring entries.

The \emph{compare} of $X$ and $Y$ with respect to an operator $\bowtie \: \in
\{=, \neq, \geq, \leq\}$, is the bit string $C$, where all entries are
$0$ and the $i$th test bit is $1$ iff $x_i \bowtie y_i$. For the
$\geq$ operator, we compute the compare as follows. Set the test bits
of $X$ by $\mid$'ing with $I_{s,f}$, then subtract $Y$, and mask out
the test bits by $\:\&\:$'ing with $I_{s,f}$. It is straightforward to
show that the $i$th test bit in the result ``survives'' the
subtraction iff $x_i \geq y_i$. The entire operation takes $O(1)$
time. We can similarly compute the compare with respect to the other
operators ($=$, $\neq$, and $\leq$) in constant time.

Given a sequence of test bits $t_1, \ldots, t_s$ stored at test bit
position in a bit string $T$, i.e., $T = t_s \cdot 0^f \ldots t_1
0^f$, the \emph{extract} of $X$ with respect to $T$, is the $f$-packed
sequence $E$ given by
\begin{equation*}
E\angle{i} =
\begin{cases}
     x_i & \text{if $t_i = 1$}, \\
      0 & \text{otherwise}.
\end{cases}
\end{equation*}
We compute the extract operation as follows. First, copy each test bit
to all positions in their field by subtracting $(I_{s,f} \gg f)$ from
$T$. Then, $\&$ the result with $X$. Again, the operation takes $O(1)$
time. We can combine the compare and extract operation to compute more
complicated operations. For instance, to compute the elementwise
maximum $M = \angle{\max(x_1, y_1),\ldots, \max(x_s,y_s)}$, compare
$X$ and $Y$ with respect to $\geq$ and let $T$ be the result. Extract
from $X$ with respect to $T$, the packed sequence $M_X$ containing
all entries in $X$ that are greater than or equal to the corresponding
entry in $Y$. Also, extract from $Y$ with respect to $\overline{T}
\:\&\: I_{r,f}$, the packed sequence $M_Y$ containing all entries in
$Y$ that are greater than or equal to the corresponding entry in
$X$. Finally, combine $M_X$ and $M_Y$ into $M$ by $\mid$'ing them.

Let $z$ be a $f$-bit integer. The \emph{rank} of $z$ in $X$, denoted
by $\rank(X,z)$, is the number of entries in $X$ smaller than or equal
to $z$. We can compute $\rank(X, z)$ in constant time as
follows. First, replicate $z$ to all fields in a words by computing $Z
= z \times 1(0^{f}1)^{s} = \angle{z, \ldots, z}$. Then, compare $X$
and $Z$ with respect to $\geq$ and store the result in a word $C$. The
number of $1$ bits in $C$ is $\rank(X,z)$. To count these, we compute
the \emph{suffix sum} of the test bits by multiplying $C$ with
$(0^{f}1)^{s}$. This produces a word $P$ such that $P\angle{i}$ is number
of test bits in the rightmost $i$ field of $C$. Finally, we extract
$P\angle{s}$ as the result. Note that the condition $f \geq \log r$ is
needed here.

All of the above $O(1)$ time algorithms, except $\rank$, are
straightforward to generalize efficiently to longer $f$-packed
sequences.  For $f$-packed sequences of length $r > s$ the time
becomes $O(r/s + 1) = O(rf/w + 1)$.

We will also need more sophisticated packed sequence
operations. First, define a \emph{$f$-packed function} of length $u$
to be a $2f$-packed sequence $G = \angle{(z_{1}, g(z_{1})), \ldots,
  (z_{u}, g(z_{u}))}$, where $z_{1} <\cdots < z_{u}$ and $g$ is any
function mapping a bit string of length $f$ to a bit string of length
$f$. The \emph{domain} of $G$, denoted $\dom(G)$, is the sequence
$\angle{z_{1}, \ldots, z_{u}}$. Let $X = \angle{x_{1}, \ldots, x_{r}}$
and $Y = \angle{y_{1}, \ldots, y_{r}}$ be $f$-packed sequences and let
$G$ be a $f$-packed function such that each entry in $X$ appears in
$\dom(G)$. Define the following operations.

 \begin{relate}
\item[$\Map(G, X):$] Return the $f$-packed sequence $\angle{g(x_1), \ldots, g(x_r)}$. 
\item[$\Lnca(X, Y):$] Return the $f$-packed sequence $\angle{\lnca(x_1, y_1), \ldots, \lnca(x_r, y_r)}$. 
\end{relate}
In other words, the $\Map$ operation applies $g$ to each entry in $X$
and $\Lnca$ is the elementwise version of the $\lnca$ operation. We
believe that an algorithm for these operations might be of independent
interest in other applications. In particular, the $\Map$ operation
appears to be a very useful primitive for algorithms using packed
sequences. Before presenting our algorithms for $\Map$ and $\Lnca$, we
show how they can be used to implement the LV-algorithm.

\section{From Landau-Vishkin to Mapping and Label Nearest Common Ancestor}
In this section we give an implementation of the LV-algorithm based on
the $\Map$ and $\Lnca$ operations. Let $P$ and $\widehat{Q}$ be
strings of length $m$ and $2m-2k$ and $k$ be an error
threshold. Recall from Section~\ref{sec:preliminaries} that an
algorithm for this case immediately generalizes to find approximate
matches in longer strings. We preprocess $P$ and $\widehat{Q}$ and
then use the constructed data structures to efficiently implement the
LV-algorithm.

\subsection{Preprocessing}
We compute the following information. Let $r = O(m)$ be the number of
diagonals in the LV-algorithm on $P$ and $\widehat{Q}$.
\begin{itemize}
\item The (generalized) suffix tree, $T_{P,\widehat{Q}}$, of $P$ and $\widehat{Q}$
  containing $O(m)$ nodes and leaves. The leaf representing suffix $i$
  in $P$ is denoted $\leaf(P,i)$, and the leaf representing suffix $j$
  in $\widehat{Q}$ is denoted $\leaf(\widehat{Q},j)$.

\item Nearest common ancestor labels for the nodes in
  $T_{P,\widehat{Q}}$ according to Theorem~\ref{thm:nca}. Hence, the
  maximum length of labels is $f = O(\log m)$. We denote the label for
  a node $v$ by $\lab(v)$.
\item The $f$-packed functions $N_{P}$, $N_{\widehat{Q}}$, and $D$,
  representing the functions given by $n_P(i) = \lab(\leaf(P,i))$, for
  $i \in \{1, \ldots, m\}$, $n_{\widehat{Q}}(j) =
  \lab(\leaf(\widehat{Q}, j))$, for $j \in \{1, \ldots, 2m-2k\}$, and
  $d(\lab(v)) = \strdepth(v)$, for any node $v$ in $T_{P,\widehat{Q}}$.
\item The $f$-packed sequences $1_{r,f}$ and $M_{r,f}$ consisting of
  $r$ copies of $1$ and $m$, respectively,  and the $f$-packed sequence
  $J_{r,f} = \angle{1, 2,\ldots, r}$.
\end{itemize}
Since $r = O(m)$, the space and preprocessing time for all of the above information is $O(m)$. 

\subsection{A Packed Landau-Vishkin Algorithm}\label{sec:reduction}
Recall that the LV-algorithm iteratively computes the sets of diagonal
positions $L_{-1}, \ldots, L_k$, where $L_e$ is the set of entries in
$L$ with error $e$. To implement the algorithm we represent each of
the sets of diagonal positions as $f$-packed sequences of length
$r$. We construct $L_{-1}$ by inserting each field in constant time
according to \eqref{LVrecursion}. After computing $L_k$, we inspect
each field in constant time and report any matches. These steps take
$O(r) = O(m)$ time in total. We show how to compute the remaining sets
of diagonal positions. Given $L_{e-1}$, $e \in \{0, \ldots, k\}$, we
compute $L_e$ as follows. First, fill in the $O(1)$ boundary fields
according to \eqref{LVrecursion1}, \eqref{LVrecursion2}, and
\eqref{LVrecursion3}. Then, compute the remaining fields using the
following $4$ steps.
\paragraph{Step 1: Compute Maximum Diagonal Positions}
Compute the $f$-packed sequence $Z$ given by 
\begin{equation*}
Z\angle{d} := \min\left(m, \max(L_{e-1}\angle{d} + 1, L_{e-1}\angle{d-1} , L_{e-1}\angle{d+1} + 1)\right).
\end{equation*} 
Thus, $Z$ corresponds to the ``$z$" part in \eqref{LVrecursion5}. We
compute $Z$ efficiently as follows. First, construct the packed
sequences $Z_1\angle{d} := L_{e-1}\angle{d} + 1$, $Z_2\angle{d} :=
L_{e-1}\angle{d-1}$, and $Z_3\angle{d} := L_{e-1}\angle{d+1} + 1$ 
by shifting and adding $1_{r,f}$. Then, compute the elementwise
maximum of $Z_1$, $Z_2$, and $Z_3$, and finally, the elementwise
minimum with $M_{r,f}$. 

\paragraph{Step 2: Translate to Suffixes}
Compute the $f$-packed sequences $Z_P$ and $Z_{\widehat{Q}}$ given by 
\begin{align*}
Z_P\angle{d} &:= Z\angle{d} + 1, \\
Z_{\widehat{Q}}\angle{d} &:= Z\angle{d} + d + m.
\end{align*}
Hence, $Z_P\angle{d}$ and $Z_{\widehat{Q}}\angle{d}$ contains the
inputs to the $\lcp$ part in \eqref{LVrecursion4}. We can compute
$Z_P$ by adding $1_{r,f}$ to $Z$ and $Z_{\widehat{Q}}$ by adding
$J_{r,f}$ and $M_{r,f}$ to $Z$. 

\paragraph{Step 3: Compute Longest Common Prefixes}
Compute the $f$-packed sequence $\LCP$ given by 
\begin{equation*}
\LCP := \Map(D, \Lnca(\Map(N_P, Z_P), \Map(N_{\widehat{Q}}, Z_{\widehat{Q}}))).
\end{equation*}
This corresponds to the computation of $\lcp$ in \eqref{LVrecursion4}. 

\paragraph{Step 4: Update State}
Finally, compute the new sequence $L_e$ of diagonal positions as 
\begin{equation*}
L_e\angle{d} = Z\angle{d} + \LCP\angle{d}.
\end{equation*}
This corresponds to the $+$ in \eqref{LVrecursion4}. 

\medskip

Steps $1$, $2$, and $4$ takes $O(rf/w + 1) = O(m\log m/w + 1)$
time. Note that a set of diagonal positions of LV-algorithm requires
$O(m\log m)$ bits to be represented. Hence, to simply output a set of
diagonal positions we must spend at least $\Omega(m\log m/w)$ time.

We parameterize the complexity for approximate string matching in
terms of the complexity for the $\Lnca$ and $\Map$ operations.
\begin{lemma}\label{lem:redux}
  Let $P$ and $Q$ be strings of length $m$ and $n$, respectively, and
  let $k$ be an error threshold. Given a data structure using $s$
  space and $p$ preprocessing time that supports $\Map$ and $\Lnca$ in
  time $q$ on $O(\log m)$-packed sequences of length $O(m)$, we can
  solve approximate string matching in time $O\left(\frac{nk}{m} \cdot
    q + \frac{nk\log m}{w} + p + n\right)$ and space $O(s + m)$.
\end{lemma}
\begin{proof}
  We consider two cases depending on $n$. First, suppose that $n \leq
  2m-2k$. Then, all of the packed sequences in the algorithm have
  length $O(m)$. Hence, we can use the data structure for $\Map$ and
  $\Lnca$ directly to implement step $3$ in time $q$. Since steps $1$,
  $2$, and $4$ use time $O(rf/w + 1) = O(m\log m/w + 1)$, we can
  compute all of the $k+1$ state transitions in time $O(k(q +
  \frac{m\log m}{w}) + m)$. With additional time and space for
  preprocessing and using the fact that $n/m = O(1)$, the result
  follows.  If $n > 2m-2k$, we apply the algorithm to $O(n/m)$
  substrings of length $2m-2k$ as described in
  Section~\ref{sec:preliminaries}. Since the computation for each of
  the substrings is independent, we can reuse space to get $O(p + m)$
  space in total. The total time is
  \begin{equation*}
    O\left(\frac{n}{m}\cdot k \cdot \left(q + \frac{m\log m}{w}\right) + p + n\right) =  O\left(\frac{nk}{m} \cdot q + \frac{nk\log m}{w} + p + n\right).
  \end{equation*}
\end{proof}

\section{Implementing $\Lnca$ and $\Map$}\label{sec:lncamap}
In this and the following section we show how to implement the
$\Lnca$ and $\Map$ operation efficiently.

For simplicity in the description of our algorithms, we will initially
assume that our word RAM model supports a constant number of
\emph{non-standard instructions}. Specifically, in addition to the
standard constant time instructions on words, e.g., arithmetic and
bitwise logical instructions, we will allow a few special constant
time instructions (the non-standard ones) defined by us. As with
standard instructions, a non-standard instruction take $O(1)$ operand
words and return $O(1)$ result words. We will subsequently implement
the non-standard instructions using either tabulation or word-level
parallelism. These two approaches lead to the two parts of
Theorem~\ref{thm:main}. We emphasize that the main result in
Theorem~\ref{thm:main} only uses standard instructions.

To implement $\Lnca$, we will simply assume that $\Lnca$ is itself
available as a non-standard instruction. Specifically, given two $f$-packed
sequences $X$ and $Y$ of length $s = \floor{w/(f+1)}$, e.g., $X$ and
$Y$ can each be stored in a single word, we can compute $\Lnca(X,Y)$
in constant time.  Since $\Lnca$ is an elementwise
operation, we immediately have the following result for general packed
sequences. 
\begin{lemma}\label{lem:nstdlnca}
  Let $X$ and $Y$ be $f$-packed sequences of length $r$. With a
  non-standard $\Lnca$ instruction, we can compute $\Lnca(X, Y)$ in
  time $O(\frac{rf}{w} + 1)$.
\end{lemma}
\begin{proof}
  Using the non-standard $\Lnca$ instruction, we compute the $i$th
  word of $\Lnca(X, Y)$ in constant time from the $i$th word of $X$
  and $Y$. Since $X$ and $Y$ are stored in $O(rf/w + 1)$ words, the
  result follows.
\end{proof}

The output words of the $\Map$ operation may depend on many words of
the input and a fast way to collect the needed information is
therefore required.  We achieve this with a number of auxiliary
operations. Let $X$ and $Y$ and be $f$-packed sequences of length $r$
and let $G$ be a $f$-packed function of length $u$. Define
\begin{relate}
\item[$\Zip(X,Y):$] Return the $2f$-packed sequence $X \ddagger Y$.
\item[$\Unzip(X \ddagger Y):$] Return $X$ and $Y$. This is the reverse of the $\Zip$ operation. 
\item[$\Merge(X, Y):$] For sorted $X$ and $Y$, return the sorted $f$-packed sequence of the $2r$ entries in $X$ and $Y$. 
\item[$\Sort(X):$] Return the $f$-packed sequence of the sorted entries in $X$.
\item[$\SortedMap(G, X):$] For sorted $X$, return $\Map(G, X)$.
\end{relate}
With these operations available as non-standard instructions, we obtain
the following result for general $f$-packed sequences.
\begin{lemma}\label{lem:nstdother}
  Let $X$ and $Y$ be $f$-packed sequences of length $r$ and let $G$ be
  a $f$-packed function of length $u$.  With $\Zip$, $\Unzip$,
  $\Merge$, $\Sort$ and $\SortedMap$ available as non-standard
  instructions, we can compute
\begin{itemize}
 \item[(i)] $\Zip(X, Y)$, $\Unzip(X \ddagger Y)$, and $\Merge(X)$ in time $O(\frac{rf}{w} + 1)$, 
 \item[(ii)] $\Sort(X)$ in time $O(\frac{rf}{w} \log r + 1)$, and 
 \item[(iii)] $\SortedMap(G,X)$ in time $O(\frac{(r+u)f}{w} + 1)$.
\end{itemize}
\end{lemma}
\begin{proof}
Let $s = \floor{w/(f+1)}$ denote the number of fields in a word.

(i) We  implement $\Zip$ and $\Unzip$ one  word at the time  as in the
algorithm for  $\Lnca$. This  takes time $O(rf/w  + 1)$.  To implement
$\Merge$,  we simulate  the standard  merge algorithm.  First,  impose a
total ordering on  the entries in $X$ and $Y$  by $\Zip$'ing them with
$J_{2r,f} = \angle{1,  \ldots, 2r}$ thus increasing the  fields of $X$
and $Y$  to $2f$ bits  (if $J_{2r,f}$ is  not available, we  can always
produce any word of it constant time by determining the leftmost entry
of the word, replicating it  to all positions, and adding the constant
word $J_{s,f}  = \angle{1, \ldots,  s}$). We compute  $\Merge(X,Y)$ in
$O(r/s)$  iterations starting  with  the smallest  fields  in $X$  and
$Y$. In each iteration, we extract  the next $s$ fields of $X$ and $Y$,
$\Merge$ them using the  non-standard instruction, and concatenate the
smallest $s$  fields $Z =  \angle{z_1, \ldots, z_s}$ of  the resulting
sequence of  length $2s$  to the  output. We then  skip over  the next
$\rank(X, z_s)$  fields of $X$ and  $\rank(Y, z_s)$ fields  of $Y$ and
continue  to  the next  iteration.  The  total  ordering ensures  that
precisely  the  output   entries  in  $Z$  are  skipped   in  $X$  and
$Y$. Finally, we $\Unzip$ the $f$  rightmost bits of each field to get
the final result. To compute $\rank$  we only need to look at the next
$s$  fields of $X$  and $Y$  and hence  each iteration  takes constant
time. In total, we use time $O(rf/w + 1)$.

(ii) We simulate the merge-sort algorithm. First, sort each of word in
$X$ using the non-standard $\Sort$ instruction. This takes $O(r/s)$
time. Starting with subsequences of length $l=s$, we repeatedly merge
pairs of consecutive subsequences into sequences of length $2l$ using
(i). After $O(\log (r/s))$ levels of recursion, we are left with a
sorted sequence. Each level takes $O(r/s + 1)$ time and hence the
total time is $O(\frac{r}{s} \cdot \log \frac{r}{s}) = O(\frac{rf}{w}
\log r)$.

(iii) We implement $\SortedMap(G,X)$ as follows. Let $G_1, \ldots,
G_{\ceil{u/s}}$ be the words of $G$. We first partition $X$ into
maximum length subsequences $X_1, \ldots, X_{\ceil{u/s}}$ such that
all entries of $X_i$ appear in $\dom(G_i)$. We do so in $\ceil{u/s}$
iterations starting with the smallest field $X$. Let $g_i$ denote the
largest field in $G_i$. In iteration $i$, we repeatedly extract the
next word from $X$ and compare the largest field of the word with
$g_i$ to identify the word of $X$ containing the end of $X_i$. Let $Z
= \angle{z_1, \ldots, z_s}$ be this word. We find the end of $X_i$ in
$Z$ by computing $h = \rank(Z, g_i)$. We concatenate each of the words
extracted and the $h$ first fields of $Z$ to form $X_i$. Finally, we proceed to
the next iteration. In total, this takes $O((r+u)/s + 1)$ time.

Next, we compute for $i = 1, \ldots, \ceil{u/s}$ the $f$-packed
sequences $\SortedMap(G_i, X_i)$ by applying the non-standard
$\SortedMap$ instruction to each word in $X_i$. Since each entry in
$X_{i}$ appears in $G_{i}$ and $X_i$ is sorted, this takes constant
time for each word in $X_i$. Finally, we concatenate the resulting
sequences into the final result. The total number of words in $X_1,
\ldots, X_{\ceil{u/s}}$ is $O((r+u)/s + 1)$ and hence the total time
is also $O((r+u)/s + 1)$.
\end{proof}

With the operations from Lemma~\ref{lem:nstdother}, we can now compute $\Map(G, X)$ as the sequence $M_2$ obtained as follows. Let $J_{r,f} = \angle{1, \ldots, r}$. 
\begin{align*}
(Z_1, Z_2) &:=\Unzip(\Sort(\Zip(X,J_{r,f})) \\
A &:= \SortedMap(G, Z_1) \\
(M_1, M_2) &:= \Unzip(\Sort(\Zip(Z_2, A))) \\
\end{align*}
We claim that $M_2 = \Map(G, X)$. Since $X$ is represented in the $f$
leftmost bits of $\Zip(X,J_{r,f}) = \angle{(x_1, 1), \ldots, (x_r,
  r)}$, we have that $\Sort(\Zip(X,J_{r,f}))$ is a $2f$-packed sequence
$\angle{(x_{i_1}, i_1), \ldots, (x_{i_r}, i_r)}$ such that $x_{i_1}
\leq \cdots \leq x_{i_r}$. Therefore, $A = \SortedMap(G, Z_1) =
\angle{g(x_{i_1}), \ldots, g(x_{i_r})}$ and hence $\Zip(Z_2, A) =
\angle{(i_1, g(x_{i_1})), \ldots, (i_r, g(x_{i_r}))}$. It follows that
$\Sort(\Zip(Z_2, A)) = \angle{(1, g(x_1)), \ldots, (r, g(x_r))}$
implying that $M_2 = \Map(G, X)$.

 We obtain the following result.
\begin{lemma}\label{lem:nstdmap}
  Let $X$ be a $f$-packed sequence of length $r$ and let $G$ be a
  $f$-packed function of length $u$ such that all entries in $X$
  appear in $\dom(G)$.  With $\Zip$, $\Unzip$, $\Merge$, $\Sort$ and
  $\SortedMap$ available as non-standard word instructions, we can
  compute $\Map(G, X)$ in time $O(\frac{(r+u)f}{w} + \frac{rf}{w} \log
  r + 1)$.
\end{lemma}
\begin{proof}
  The above algorithm requires $2$ $\Sort$, $\Zip$, and $\Unzip$
  operations on packed sequences of length $r$ and a $\SortedMap$
  operation on a packed function of length $u$ and a packed sequence
  of length $r$. By Lemma~\ref{lem:nstdother} and the observation from
  the proof of Lemma~\ref{lem:nstdother}(i) that we can compute
  $J_{r,f}$ in constant time per word, we compute $\Map(G,X)$ in time
  $O(\frac{(r+u)f}{w} + \frac{rf}{w} \log r + 1)$.
\end{proof}



By a standard tabulation of the non-standard instructions, we obtain algorithms for $\Lnca$ and $\Map$ which in turn provides us with Theorem~\ref{thm:main}(i). 
\begin{theorem}\label{thm:tabulation}
Approximate string matching for strings $P$ and $Q$ of length $m$ and $n$, respectively, with error threshold $k$, can be solved in time $O(nk\cdot\frac{\log^2 m}{\log n} + n)$ and space $O(n^\epsilon + m)$, for any constant $\epsilon > 0$. 
\end{theorem}
\begin{proof}
  Modify the $f$-packed sequence representation to only fill up the $b
  = \delta \log n$ leftmost bits of each words, for some constant
  $\delta > 0$. Implement the standard operations in all our packed
  sequence algorithms as before and for the non-standard instructions
  $\Lnca$, $\Zip$, $\Unzip$, $\Sort$, $\Merge$, and $\SortedMap$
  construct lookup tables indexed by the inputs to the operation and
  storing the corresponding output. Each of the $2^{O(b)}$ entries the
  lookup tables stores $O(b) = O(w)$ bits and therefore the space for
  the tables is $2^{O(b)} = n^{O(\delta)}$. It is straightforward to
  compute each entry in time polynomial in $b$ and therefore the total
  preprocessing time is also $2^{O(b)} b^{O(1)} = n^{O(\delta)}$. For
  any constant $\epsilon > 0$, we can choose $\delta$ such that the
  total preprocessing time and space is $O(n^\epsilon)$.

  We can now implement $\Lnca$ and $\Map$ according to
  Lemma~\ref{lem:nstdlnca} and \ref{lem:nstdmap} with $w = b = O(\log
  n)$ without the need for non-standard instruction in time
  $O(\frac{rf}{\log n} + 1)$ and $O(\frac{(r+u)f}{\log n} +
  \frac{rf}{\log n} \log r + 1)$, respectively. We plug this into the
  reduction of Lemma~\ref{lem:redux}. We have that $r,u = O(m)$ and $f
  = O(\log m)$ and therefore $q = O(\frac{(r+u)f}{\log n} +
  \frac{rf}{\log n} \log r + 1) = O(\frac{m\log^2 m}{\log n} +
  1)$. Since $s = p = O(n^\epsilon)$, we obtain an algorithm for
  approximate string matching using space $O(n^\epsilon + m)$ and time
  $O(\frac{nk}{m} \cdot \frac{m\log^2 m}{\log n} + n) = O(nk \cdot
  \frac{\log^2 m}{\log n} + n)$.
\end{proof}

\section{Exploiting Word-Level Parallelism}\label{sec:word}
For part (ii) of Theorem~\ref{thm:main} we implement each of the
non-standard instructions $\Zip$, $\Unzip$, $\Sort$, $\Merge$,
$\SortedMap$, and $\Lnca$ using only the standard arithmetic and
bitwise instruction of the word RAM. This allows us to take full
advantage of long word lengths. Furthermore, this also gives us a more
space-efficient algorithm than the one above since no lookup tables
are needed. In the following sections, we present algorithms for each
of the non-standard instructions and use these to derive efficient
algorithms for the $f$-packed sequence operations. The results for
$\Zip$, $\Unzip$ and $\Merge$ are well-known and the result for
$\Sort$ is a simple extension of $\Merge$. The results for
$\SortedMap$ and $\Lnca$ are new. Throughout this section, let $s =
\floor{w/(f+1)}$ denote the number of fields in a word, and assume
without loss of generality that $s$ is a power of $2$.

\subsection{Zipping and Unzipping}
We present an $O(\log s)$ algorithm for the $\Zip$ instruction based on
the following recursive algorithm. Let $X = \angle{x_1, \ldots, x_s}$
and $Y = \angle{y_1, \ldots, y_s}$ be $f$-packed sequences. If $s=1$
return $x_1 \cdot y_1$. Otherwise, recursively  compute the packed
sequence
\begin{equation*}
\left(\angle{x_{s/2 + 1}, \ldots, x_s} \ddagger  \angle{y_{s/2 + 1}, \ldots, y_s}\right) \cdot  \left(\angle{x_1, \ldots, x_{s/2}} \ddagger  \angle{y_1, \ldots, y_{s/2}}\right) .
\end{equation*}
It is straightforward to verify that the returned sequence is $X
\ddagger Y$. We implement each level of the recursion in parallel. Let
$Z = Y \cdot X = \angle{x_1, \ldots, x_s, y_1, \ldots, y_s}$. The
algorithm works in $\log s$ steps, where each step corresponds to a
recursion level. At step $i$, $i = 1, \ldots, \log s$, $Z$ consists of
$2^i$ subsequences of length $2^{\log s - i+1}$ stored in consecutive
fields. To compute the packed sequence representing level $i+1$, we
extract the middle $2^{\log s - i}$ fields of each of the $2^i$
subsequences and swap their leftmost and rightmost halves. Each step
takes $O(1)$ time and hence the algorithm uses time $O(\log s)$. To
implement $\Unzip$, simply we carry out the steps in reverse.

This leads to the following result for general $f$-packed sequences.
\begin{lemma}\label{lem:packedzip}
For $f$-packed sequences $X$ and $Y$ of length $r$ we can compute $\Zip(X,Y)$ and $\Unzip(X \ddagger Y)$ in time $O(\frac{rf}{w}\log w + 1)$.
\end{lemma}  
\begin{proof}
Apply the algorithm from the proof of Lemma~\ref{lem:nstdother}(i)
using the $O(\log s)$ implementation of the non-standard $\Zip$ and
$\Unzip$ instructions. The time is $O(r \log s/s + 1) = O(rf\log w/w +
1)$.
\end{proof}

\subsection{Merging and Sorting}
We review an $O(\log s)$ algorithm for the $\Merge$ instruction due to
Albers and Hagerup~\cite{AH1997} and subsequently extend it to an
$O(\log^2 s)$ algorithm for the $\Sort$ instruction. Both results are
based on a fast implementation of \emph{bitonic sorting}, which we
review first.

\subsubsection{Bitonic Sorting} 
A $f$-packed sequence $Z = \angle{z_1, \ldots, z_s}$ is \emph{bitonic}
if 1) for some $i$, $1 \leq i \leq s$, $z_1, \ldots, z_i$ is a
non-decreasing sequence and $z_{i+1}, \ldots, z_s$ is a non-increasing
sequence, or 2) there is a cyclic shift of $Z$ such that 1)
holds. Batcher~\cite{Batcher1968} gave the following recursive
algorithm to sort a bitonic sequence.  Let $Z = \angle{z_1, \ldots,
  z_s}$ be a $f$-packed bitonic sequence. If $s = 1$ we are
done. Otherwise, compute and recursively sort the sequences
\begin{align*}
 Z_{\min} &= \min(z_1, z_{1+s/2}), \min(z_2, z_{2 + s/2}), \ldots, \min(z_{s/2}, z_s) \\
 Z_{\max} &= \max(z_1, z_{1+s/2}), \max(z_2, z_{2+s/2}), \ldots, \max(z_{s/2}, z_s)
\end{align*}
and return $Z_{\max} \cdot Z_{\min}$. For a proof of correctness, see
e.g. \cite[chap. 27]{CLRS2001}. Note that it suffices to show that
$X_{\min}$ and $X_{\max}$ are bitonic sequences and that all values in
$X_{\min}$ are smaller than all values in $X_{\max}$.

Albers and Hagerup~\cite{AH1997} gave an $O(\log s)$ algorithm using
an idea similar to the above algorithm for $\Zip$. The algorithm works
in $\log s+1$ steps, where each step corresponds to a recursion
level. At step $i$, $i=0, \ldots, \log s$, $Z$ consists of $2^i$
bitonic sequences of length $2^{\log s - i}$ stored in consecutive
fields. To compute the packed sequence representing level $i+1$, we
extract the leftmost and rightmost halves of each of $2^i$ bitonic
sequences, compute their elementwise minimum and maximum, and
concatenate the results. Each step takes $O(1)$ time and hence the
algorithm uses time $O(\log s)$.

\subsubsection{Merging}\label{sec:merging}
Let $X = \angle{x_1, \ldots, x_s}$ and $Y = \angle{y_1, \ldots, y_s}$
be sorted $f$-packed sequence. To implement $\Merge(X,
Y)$, we compute the reverse of $Y$, denoted by $Y^R = \angle{y_s,
  \ldots, y_1}$, and then apply the bitonic sorting algorithm to $Y^R
\cdot X$. Since $X$ and $Y$ are sorted, the sequence $X \cdot Y^R$ is
bitonic and hence the algorithm returns the sorted sequence of the
entries from $X$ and $Y$. Given $Y$, it is straightforward to compute
$Y^R$ in $O(\log s)$ time using the property that $Y^R = \angle{y_{1+
    s/2}, \ldots, y_{s}}^R \cdot \angle{y_1, \ldots, y_s/2}^R$ and a
parallel recursive algorithm similar to the algorithms for $\Zip$ and
$\Merge$. Hence, the algorithm for $\Merge$ uses $O(\log s)$ time.

This leads to the following result for general $f$-packed sequences.
\begin{lemma}[Albers and Hagerup~\cite{AH1997}]\label{lem:packedmerge}
For $f$-packed sequences $X$ and $Y$ of length $r$, we can compute $\Merge(X,Y)$ in time $O(\frac{rf}{w}\log w + 1)$.
\end{lemma}  
\begin{proof}
  Apply the algorithm from the proof of Lemma~\ref{lem:nstdother}(i)
  using the $O(\log s)$ implementation of the $\Merge$
  instruction. The time is $O(r \log s/s + 1) = O(rf\log w/w + 1)$.
\end{proof}

\subsubsection{Sorting}
Let $X = \angle{x_1, \ldots, x_s}$ be a $f$-packed sequence. We give
an $O(\log^2 s)$ algorithm for $\Sort(X)$. Starting from subsequences
of length $1$, we repeatedly merge subsequences until we have a single
sorted sequence. The algorithm works in $\log s + 1$ steps. At step
$i$, $i=\log s, \ldots, 0$, $X$ consists of $2^i$ sorted sequences of
length $2^{\log s - i}$ stored in consecutive fields. Note that the
steps here are numbered in decreasing order. To compute the packed
sequence representing level $i-1$, we merge pairs of adjacent
sequences by reversing the rightmost one of each pair and sorting the
pair with a bitonic sort. At level $i$, the reverse and bitonic sort takes
$O(\log i)$ time using the algorithms described above. Hence, the
algorithm for $\Sort(X)$ uses time $O(\sum_{i=0}^{\log s} \log i) =
O(\log^2 s)$.

This leads to the following result for general $f$-packed sequences.
\begin{lemma}\label{lem:packedsort}
For a $f$-packed sequence $X$ of length $r$, we can compute $\Sort(X)$ in time $O(\frac{rf}{w}\log r \log w + 1)$.
\end{lemma}
\begin{proof}
  We implement the merge-sort algorithm as in the proof of
  Lemma~\ref{lem:nstdother}(ii). Sorting all words takes time
  $O(\frac{r}{s} \log^2 s + 1) = O(\frac{rf}{w} \log^2 w + 1)$ with
  the $O(\log^2 s)$ implementation of the $\Sort$ instruction. Each of
  the $O(\log (r/s)) = O(\log r)$ $\Merge$ steps takes time
  $O(\frac{rf}{w}\log w + 1)$ by Lemma~\ref{lem:packedmerge}. In
  total, $\Sort(X)$ takes time $O(\frac{rf}{w}\log^2 w +
  \frac{rf}{w}\log r \log w + 1) = O(\frac{rf}{w}\log r \log w + 1)$.
\end{proof}

\subsection{Mapping}
We present an $O(\log s)$ algorithm for the $\SortedMap$
instruction. Our algorithm uses a fast algorithm to \emph{compact}
packed sequences by Andersson et al.~\cite{AHNR1998}, which we review
first.

\subsubsection{Compacting}
Let $X = \angle{x_1, \ldots, x_s}$ be a $f$-packed sequence. We
consider field $i$ with test bit $t_i$ in $X$ to be \emph{vacant} if
$t_i = 1$ and \emph{occupied} otherwise. If $X$ contains $l$ occupied
fields, the \emph{compact} operation on $X$ returns a $f$-packed
sequence $C$ consisting of the occupied fields of $X$ tightly packed
in the $l$ rightmost fields of $A$ and in the same order as they appear in
$X$. Andersson et al.~\cite[Lemma 6.4]{AHNR1998} gave an $O(\log s)$
algorithm to compact $X$. The algorithm first extracts the test bits
and computes their prefix sum in a $f$-packed sequence $P$. Thus, 
$P\angle{i}$ contains the number of fields $X\angle{i}$ needs to be
shifted to the right in the final result. Note that the number of
vacant positions in $P$ can be up to $s$ and hence we need $f \geq
\log s$. We then move the occupied fields in $X$ to their correct
position in $\log s+1$ steps. At step $i$, $i=0, \ldots, \log s$, 
extract all occupied ﬁelds $j$ from $X$ such that bit $i$ of
$P\angle{j}$ is $1$. Move these fields $2^i$ position to the right and
insert them back into $X$.

The algorithm moves the occupied fields their correct position
assuming that no fields ``collide'' during the movement. For a proof
of this fact, see~\cite[Section 3.4.3]{Leighton1992}. Each step of the
movement takes constant time and hence the total running time is
$O(\log s)$. Thus, we have the following result.

\begin{lemma}[Andersson et al.~\cite{AHNR1998}]\label{lem:packedcompact}
  We can compact a $f$-packed sequence of length $s$ stored in $O(1)$
  words in time $O(\log s)$.
\end{lemma}  

\subsubsection{Mapping}
Let $X = \angle{x_1, \ldots, x_s}$ be a sorted $f$-packed sequence and
let $G = \angle{(z_1, g(z_{1})), \ldots, (z_s, g(z_{1}))}$ be a
$f$-packed function representing a function such that all entries in
$X$ appear in $\dom(G)$. We compute $\SortedMap(G,X)$ in $4$ steps:

\paragraph{Step 1: Merge Sequences} First, construct $2f+1$-packed
sequences $\widehat{X} = \angle{(x_1, 0, 0), \ldots, (x_s, 0, 0)}$ and
$\widehat{G} = \angle{(z_1, 1, g(z_1)), \ldots, (z_s, 1, g(z_s)}$ with
two zips. The $1$-bit subfield in the middle, called the \emph{origin
  bit}, is $0$ for $\widehat{X}$ and $1$ for $\widehat{G}$.

Compute $M = \Merge(\widehat{G}, \widehat{X})$. Since entries from $X$
and $\dom(G)$ appear in the rightmost $f$-bits of the fields in
$\widehat{G}$ and $\widehat{X}$, identical values from $X$ and
$\dom(G)$ are grouped together in $M$. We call each such a group a
\emph{chain}. Since the entries in $\dom(G)$ are unique and all
entries in $X$ appears in $\dom(G)$, each chain contains one entry
from $\widehat{G}$ followed by $0$ or more entries from
$\widehat{X}$. Furthermore, since the origin bit is $1$ for entries
from $\widehat{G}$ and $0$ from $\widehat{X}$, each chain starts with
a field from $\widehat{G}$. Thus, $M$ is the concatenation of
$|\dom(F)| = s$ chains:
\begin{equation*}
  M = C_1 \cdots C_s 
  = \angle{(z_1, 1, g(z_1)), \underbrace{(z_1, 0, 0), \ldots, (z_1,0,0)}_{\text{$0$ or more}} } \cdots \angle{(z_s, 1, g(z_s)), \underbrace{(z_s, 0, 0), \ldots, (z_s,0,0)}_{\text{$0$ or more}}}. 
\end{equation*}
All operations in step $1$ takes $O(1)$ time except for $\Merge$ that
takes $O(\log s)$ time using the algorithm from
Section~\ref{sec:merging}.

Consider a chain $C = \angle{(z_{j}, 1, f(z_{j})), (z_{j}, 0, 0),
  \ldots, (z_{j},0,0)}$ in $M$ with $p$ fields. Each of the $p-1$
fields $(z_{j}, 0, 0), \ldots, (z_{j},0,0)$ correspond to $p-1$
identical fields from $X$, and should therefore be replaced by $p-1$
copies of $f(z_j)$ in the final result (note that for $p=1$, $z_j \not
\in X$ and therefore $f(z_j)$ is not present in the final result). The
following $3$ steps convert $C$ to $p-1$ copies of $f(z_j)$ as
follows. Step $2$ removes the leftmost field of $C$. If $p=1$, $C =
\angle{(z_j, 1, f(z_j))}$ is completely removed and does not
participate further in the computation. Otherwise, we are left with
$C=\angle{(z_{j}, 1, f(z_{j})), (z_{j}, 0, 0), \ldots, (z_{j},0,0)}$
with $p-1 > 0$ fields. Step $3$ computes the chain lengths and
replaces $C$ with $\angle{(p-1, 1, f(z_j)}$. Finally, step $4$
converts this to $p-1$ copies of $f(z_j)$.
  
\paragraph{Step 2: Reduce Chains}
Extract the origin bits from $M$ into a sequence $O$. Shift $O$ to the
right to set all entries to right of the start of each chain to be
vacant and then compact.  The resulting sequence $M^1$ is a
subsequence of $l$ reduced chains $C_{i_1}, \ldots, C_{i_l}$ from $C_1
\cdots C_r$. Note that $l$ is the number of chains of length $>1$ in
$M$ and therefore the number of unique entries in $X$. Hence,
\begin{equation*}
  M^1 = C_{i_1} \cdots C_{i_l} 
  = \angle{(z_{i_1}, 1, f(z_{i_1})), \underbrace{(z_{i_1}, 0, 0), \ldots, (z_1,0,0)}_{\text{$0$ or more}} } \cdots \angle{(z_{i_l}, 1, f(z_{i_l})), \underbrace{(z_{i_l}, 0, 0), \ldots, (z_{i_l},0,0)}_{\text{$0$ or more}}}.   
\end{equation*}
All operations in step $2$ takes $O(1)$ time except for the compact
operation that takes $O(\log s)$ time by
Lemma~\ref{lem:packedcompact}.

\paragraph{Step 3: Compute Chain Lengths} Replace the rightmost
subentry of each field in $M^1$ by the index of the field. To do so
unzip the rightmost subentry and zip in the sequence $J_{r,f}$
instead. Set all fields with origin bit $0$ to be vacant producing a
sequence $M^s$ given by
 \begin{equation*}
 M^s = \angle{(s(C_{i_1}), 1, f(z_{i_1})), \underbrace{\bot, \ldots, \bot}_{\text{$0$ or more}} } \cdots \angle{s(C_{i_l}), 1, f(z_{i_l})), \underbrace{\bot, \ldots, \bot}_{\text{$0$ or more}} },
 \end{equation*}
 where $s(C)$ is the start index of chain $C$ and $\bot$ denotes a
 vacant field. We compact $M^s$ and unzip the origin bits to get a
 $2f$-packed sequence
\begin{equation*}
   S = \angle{(s(C_{i_1}), f(z_{i_1})), \ldots, (s(C_{i_l}), f(z_{i_l}))}.
\end{equation*}
The length of $C_{i_j}$, denoted $l(C_{i_j})$, is given by $l(C_{i_j})
= s(C_{i_{j+1}}) - s(C_{i_{j}})$, $1 \leq j < l$. Hence, we can
compute the lengths for all chains except the $C_{i_l}$ by subtracting
the rightmost subentries of $S$ from the rightmost subentries of $S$
shifted to the right by one field. We compute the length of $C_{i_l}$
as $|S| - s(C_{i_l}) + 1$ and store all lengths as the $f$-packed
sequence
\begin{equation*}
   L = \angle{(l(C_{i_1}), f(z_{i_1})), \ldots, (l(C_{i_l}), f(z_{i_l}))}.
\end{equation*}
As in step $2$, all operations in step $3$ takes $O(1)$ time except for the compact operation that takes $O(\log s)$ time by Lemma~\ref{lem:packedcompact}.  

\paragraph{Step 4: Copy Function Values} Expand each field
$(l(C_{i_j}), f(z_j))$ in $L$ to $l(C_{i_j})$ copies of $f(z_j)$. To
do so, we run a reverse version of the compact algorithm that copies
fields whenever fields are moved. We copy the fields in $\log s$
iterations. At iteration $h$, $h = \log s, \ldots, 0$ extract all
fields $j$ from $X$ such that bit $h$ of the right subentry of
$L\angle{j}$ is $1$. Replicate each of these fields to the $2^h$
fields to their left. Finally, we unzip the rightmost subentry to get
the final result.
Each of the $O(\log s)$ iterations take $O(1)$ time and therefore step $4$ takes $O(\log s)$ time. 

\medskip

Each step of the algorithm for $\SortedMap(F,X)$ uses time $O(\log s)$. This leads to the following result for general $f$-packed sequences. 
\begin{lemma}\label{lem:packedsortedmap}
For a sorted $f$-packed sequence $X$ with $r$ entries and a $f$-packed
function $G$ with $u$ entries such that all entries in $X$ appear in
$\dom(G)$, we can compute $\SortedMap(G,X)$ in time
$O(\frac{(r+u)f}{w}\log w + 1)$. 
\end{lemma}
\begin{proof}
Each of the $O((r+u)/s) = O((r+u)f/w)$ $\SortedMap$ instructions used
in the algorithm in the proof of Lemma~\ref{lem:nstdother} take
$O(\log s) = O(\log w)$ time. In total, the algorithm takes time
$O(\frac{(r+u)f}{w}\log w + 1)$.
\end{proof}

Plugging the above results in the algorithm for $\Map$ from
Section~\ref{sec:lncamap}, we obtain the following result.
\begin{lemma}\label{lem:packedmap}
  For a $f$-packed sequence $X$ with $r$ entries and a $f$-packed
  function $G$ with $u$ entries such that all entries in $X$ appear in
  $\dom(G)$, we can compute $\Map(G,X)$ in time $O(\frac{rf}{w}\log r
  \log w + \frac{(r+u)f}{w}\log w + 1)$.
\end{lemma}
\begin{proof}
The algorithm from Section~\ref{sec:lncamap} does a constant number of
$\Sort$, $\Zip$, and $\Unzip$ operations on packed sequences of length
$r$ and performs a single $\SortedMap$ operation on a packed function
of length $u$ and a sequence of length $r$. By
Lemmas~\ref{lem:packedzip}, \ref{lem:packedsort}, and
\ref{lem:packedsortedmap} this takes time $O(\frac{rf}{w}\log r \log w
+ \frac{(r+u)f}{w}\log w + 1)$.
\end{proof}

\subsection{Label Nearest Common Ancestor}
We present an $O(\log f)$ algorithm for the $\Lnca$ instruction. We first review the relevant features of the labeling scheme from Alstrup et al.~\cite{AGKR2004}.

\subsubsection{The Labeling Scheme}\label{sec:labelingscheme}
Let $T$ a tree with $t$ nodes. The labeling scheme from Alstrup et
al.~\cite{AGKR2004} assigns to each node $v$ in $T$ a unique bit
string, called the \emph{label} and denoted $\lab(v)$, of length $O(\log
t)$ bits. The label is the concatenation of three identical length bit
strings:
$$
\lab(v) = p(v) \cdot b(v) \cdot l(v)
$$
The label $p(v)$, called the \emph{part label}, is the concatenation
of an alternating sequence of variable length bit strings called
\emph{lights parts} and \emph{heavy parts}:
\begin{equation*}
p(v) = h_0 \cdot l_1 \cdot h_1 \cdots  l_j \cdot h_{j}
\end{equation*}
Each heavy and light part in the sequence identify special nodes on
the path from the root of $T$ to $v$. The leftmost part, $h_0$,
identifies the root. The total number of parts in $p(v)$ and the total
length of the parts is at most $O(\log t)$. For simplicity in our
algorithm, we use a version of the labeling scheme where the parts are
constructed using \emph{prefix free codes} (see remark 2 in Section 5
of Alstrup et al.~\cite{AGKR2004}). This implies that if part labels
$p(v)$ and $p(w)$ agree on the leftmost $i-1$ parts, then part $i$ in
$p(v)$ is not a prefix of part $i$ in $p(w)$ and vice versa. We also
prefix all parts in all part labels by a single $0$ bit. This
increases the minimum length of a part to $2$ and ensures the longest
common prefix of any two parts is at least $1$. Since the total number
of parts in a part label is $O(\log t)$, this increases the total
length of part labels by at most a factor $2$.

The sublabels $b(v)$ and $l(v)$ identify the boundaries of parts in
$p(v)$. The sublabel $b(v)$ has length $|p(v)|+ 1$ and is $1$ at each
leftmost position of a light or heavy part in $p(v)$ and $1$ at
position $|p(v)|+1$. The sublabel $l(v)$ has length $|p(v)|$ and is
$1$ at each leftmost position of a light part in $p(v)$. The total
length of $\lab(v)$ is $3|p(v)| + 1 = O(\log t)$.

For our purposes, we need to store labels from $T$ in equal length
fields in packed sequences. To do so compute the length $c$ of the
maximum length part label assigned to a node in $T$. Note that $c$ is
an upper bound on any sublabel in $T$. We store all labels in fields of
length $f = 3c$ bits of the form $(p(v) \cdot 0^{c - |p(v)|}, b(v)
\cdot 0^{c - |b(v)|}, l(v) \cdot 0^{c - |l(v)|})$, i.e., each sublabel
is stored in a subfield of length $c$ aligned to the left of the
subfield and padded with $0$'s to the right.

Alstrup et al.~\cite{AGKR2004} showed how to compute $\lnca$ of two
labels in $T$. We restate it here in an form suitable for our
purposes. First we need some definitions. For two bit strings $x$ and
$y$, we write $x <_{\lex} y$ if and only if $x$ precedes $y$ in the
\emph{lexicographic} order on binary strings, that is, $x$ is a prefix
of $y$ or the first bit in which $x$ and $y$ differ is $0$ in $x$ and
$1$ in $y$. To compute the lexicographic minimum of $x$ and $y$,
denoted $\min_{\lex}$, we can shift the smaller to left align $x$ and
$y$ and then compute the numerical minimum. Let $x = p(v)$ and $y =
p(w)$ be part labels of nodes $v$ and $w$. The \emph{longest common
  part prefix} of $x$ and $y$, denoted $\lcpp(x,y)$, is the longest
common prefix of $x$ and $y$ that ends at a part boundary. The
\emph{leftmost distinguishing part} of $x$ with respect to $y$,
denoted $\ldp_y(x)$, is the part in $x$ immediately to the right of
$\lcpp(x,y)$.

\begin{lemma}[Alstrup et al.~\cite{AGKR2004}(Lemma 5)]\label{lem:nca}
Let $x = p(v)$ and $y = p(w)$ be part labels of nodes $v$ and $w$. Then,  
\begin{equation*}
p(\nca(v,w)) = 
\begin{cases}
    \lcpp(x,y)  & \text{if $\ldp_y(x)$ is a heavy part }, \\
    \lcpp(x,y) \mid (\min_\lex(\ldp_y(x), \ldp_x(y)) \gg |\lcpp(x,y)|)  & \text{if $\ldp_y(x)$ is a light part}.
\end{cases}
\end{equation*}
\end{lemma}
From the information in the label and Lemma~\ref{lem:nca} it is straightforward to compute $\lnca(x,y)$ for any two labels $x,y$ stored in $O(1)$ words in $O(1)$ time using straightforward bit manipluations. We present an elementwise version for packed sequences in the following section. 

\subsubsection{Computing Label Nearest Common Ancestor}
Let $X$ and $Y$ be $f$-packed sequences of length $s$. We present an
$O(\log f)$ algorithm for the $\Lnca(X,Y)$ instruction. We first need
some additional useful operations. Let $x \neq 0$ be a bit
string. Define $\lmb(x)$ and $\rmb(x)$ to be the position of the
leftmost and rightmost $1$ bit of $x$, respectively. Define
\begin{align*}
    \lsmear(x) &= 0^{|x|-\rmb(x)}\cdot 1^{\rmb(x)}  \\
    \rsmear(x) &= 1^{\lmb(x)} \cdot 0^{|x|-\lmb(x)}
\end{align*}
Thus, $\lsmear(x)$ ``smears'' the rightmost $1$ bit to the right and
clears all bits to left. Symmetrically, $\rsmear(x)$ smears the
leftmost $1$ bit to the left and clears all bits to left. We can
compute $\lsmear(x)$ in $O(1)$ time since $\lsmear(x) = x \oplus
(x-1)$ (see e.g. Knuth~\cite{Knuth2008}). Since $\rsmear(x) =
(\lsmear(x^R))^R$ and a reverse takes time $O(\log |x|)$ (as described
in Section~\ref{sec:merging}) we can compute $\rsmear(x)$ in time
$O(\log |x|)$. Elementwise versions of $\lsmear$ and $\rsmear$ on
$f$-packed sequences are easy to obtain. Given a $f$-packed sequence
$X$ of length $s$, we can compute the elementwise $\lsmear$ as $X
\oplus (X-1_{s,f})$. We can reverse all fields in time $O(\log f)$ and
hence we can compute the elementwise $\rsmear$ in time $O(\log f)$.

We compute $\Lnca(X,Y)$ as follows. We handle identical pairs of
labels first, that is, we extract all fields $i$ from $X\angle{i}$
such that $X\angle{i} = Y\angle{i}$ into a sequence $L'$. Since
$\lnca(x,x) = x$ for any $x$, we have that $\Lnca(X,Y)\angle{i} =
X\angle{i}$ for these fields. We handle the remaining fields using the
$3$ step algorithm below. We then $\mid$ the result with $L'$ to get
the final sequence.

\paragraph{Step 1: Compute Masks}
Unzip the $f/3$-packed sequences $X_p$, $X_b$, $X_l$, $Y_p$, $Y_b$,
and $Y_l$ from $X$ and $Y$ corresponding to each of the $3$
sublabels. We compute $f/3$-packed sequences of masks $Z$, $M$, $M_X$,
and $M_Y$ to extract relevant parts from $X_p$ and $Y_p$. The mask are
given by
\begin{align*}
Z\angle{i} &:= \lsmear(X_p\angle{i} \oplus Y_p\angle{i})\\
U\angle{i} &:= \rsmear(X_b\angle{i} \:\&\: Z\angle{i}) \\
R_Y\angle{i} &:= \lsmear((U\angle{i} \gg 1) \:\&\: X_b\angle{i}) \ll 1\\
R_X\angle{i} &:= \lsmear((U\angle{i} \gg 1) \:\&\: Y_b\angle{i}) \ll 1
\end{align*}
We explain the contents of the masks in the
following. Figure~\ref{fig:lnca} illustrates the computations.
\begin{figure}[t]
  \centering \includegraphics[scale=.5]{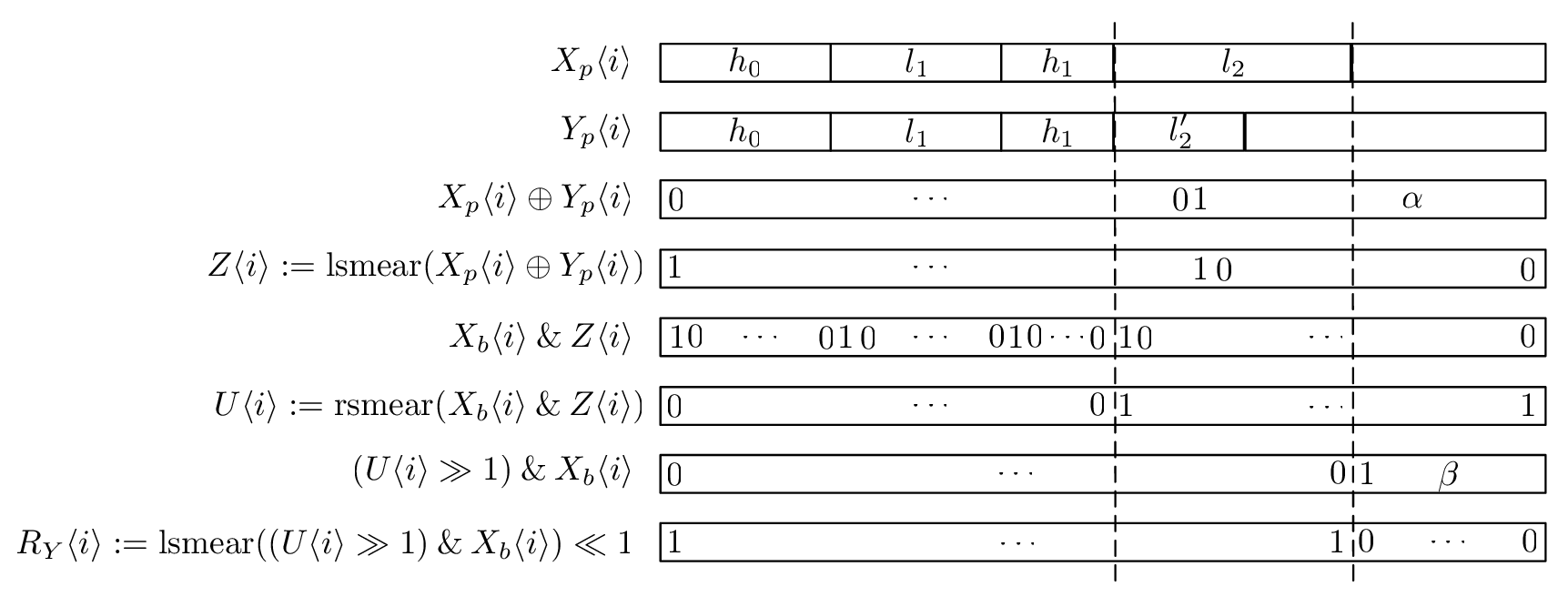}
  \caption{Computing $\lnca(X_p\angle{i}, Y_p\angle{i})$. The solid
    lines in $X_p\angle{i}$ and $Y_p\angle{i}$ show part boundaries
    and the dashed lines show boundaries for $\lcpp(X_p\angle{i},
    Y_p\angle{i})$ and $\ldp_{Y_p\angle{i}}(X_p\angle{i})$. $\alpha$
    and $\beta$ are arbitrary bit strings.}
   \label{fig:lnca}
\end{figure}
The mask $Z\angle{i}$ consists of $1$'s in position $z =
\rmb(X_p\angle{i} \oplus Y_p\angle{i})$ and all positions to the left
of $z$. Since $X_p$ and $Y_p$ are distinct labels, $z$ is the rightmost
position where $X_p\angle{i}$ and $Y_p\angle{i}$ differ. Since the
parts are prefix free encoded and prefixed with $0$, we have that $z$
is a position within $\ldp_{Y_p\angle{i}}(X_p\angle{i})$ and
$\ldp_{X_p\angle{i}}(Y_p\angle{i})$ and it is not the leftmost
position.  Consequently, $u = \lmb(X_b\angle{i} \:\&\: Z\angle{i})$ is
the leftmost position of $\ldp_{Y_p\angle{i}}(X_p\angle{i})$ and
$\ldp_{X_p\angle{i}}(Y_p\angle{i})$, and therefore the leftmost
position to the right of $\lcpp(X_p\angle{i}, Y_p\angle{i})$. Hence,
$U\angle{i} = \rsmear(X_b\angle{i} \:\&\: Z\angle{i})$ consists of
$1$'s in all positions to the right of $\lcpp(X_p\angle{i},
Y_p\angle{i})$.  This implies that $\lmb(U\angle{i} \gg 1) \:\&\:
X_b\angle{i})$ is the position immediately to the right of
$\ldp_{Y_p\angle{i}}(X_p\angle{i})$ (if
$\ldp_{Y_p\angle{i}}(X_p\angle{i})$ is the rightmost part this still
holds due to the extra bit in $X_b$ at the rightmost
position). Therefore $R_Y\angle{i} := \lsmear((U\angle{i} \gg 1)
\:\&\: X_b\angle{i}) \ll 1$ consists of $1$'s in all positions of
$\lcpp(X_p\angle{i}, Y_p\angle{i})$ and
$\ldp_{Y_p\angle{i}}(X_p\angle{i})$. Symmetrically, $R_X\angle{i} :=
\lsmear((U\angle{i} \gg 1) \:\&\: Y_b\angle{i}) \ll 1$ consists of
$1$'s in all positions of $\lcpp(X_p\angle{i}, Y_p\angle{i})$ and
$\ldp_{X_p\angle{i}}(Y_p\angle{i})$.

All operations except the elementwise $\rsmear$ in the computation of
$U$ are straightforward to compute in $O(1)$ time. Hence, the time for
this step is $O(\log f)$.

\paragraph{Step 2: Extract Relevants Parts}
Compute the $f/3$-packed sequences $\LCPP$, $\LDP_Y$, $\LDP_X$, and $M$ given by
\begin{align*}
\LCPP\angle{i} &:= X_p\angle{i} \:\&\: \overline{U\angle{i}} \\
\LDP_Y\angle{i} &:= X_p\angle{i} \:\&\: R_Y\angle{i} \:\&\: U\angle{i} \\
\LDP_X\angle{i} &:= Y_p\angle{i} \:\&\: R_X\angle{i} \:\&\: U\angle{i}\\
M\angle{i} &:= \min(\LDP_Y\angle{i}, \LDP_X\angle{i})
\end{align*}
From the definition of the mask in step $1$, we have that
$\LCPP\angle{i} = \lcpp(X_p\angle{i}, Y_p\angle{i})$. The sequence
$\LDP_Y\angle{i}$ is $X\angle{i}$ where all but
$\ldp_{Y_p\angle{i}}(X_p\angle{i})$ is zeroed and therefore
$\LDP_Y\angle{i} = \ldp_{Y_p\angle{i}}(X_p\angle{i}) \gg
|\lcpp(X_p\angle{i}, Y_p\angle{i})|$ (see
Figure~\ref{fig:lnca}). Similarly, $\LDP_X\angle{i} =
\ldp_{X_p\angle{i}}(Y_p\angle{i}) \gg |\lcpp(X_p\angle{i},
Y_p\angle{i})|$. The parts $\ldp_{Y_p\angle{i}}(X_p\angle{i})$ and
$\ldp_{X_p\angle{i}}(Y_p\angle{i})$ are left aligned in
$\LDP_Y\angle{i}$ and $\LDP_X\angle{i}$ and all other positions are
$0$. Hence,
\begin{equation*}
M\angle{i} = \min(\LDP_Y\angle{i}, \LDP_X\angle{i}) = \min_{\lex}(\LDP_Y\angle{i}, \LDP_X\angle{i}) \gg |\lcpp(X_p\angle{i}, Y_p\angle{i})| .
\end{equation*}
The time for this step is $O(1)$. 
\paragraph{Step 3: Construct Labels}
The part labels are computed as the $f/3$-packed sequence $P$ given by 
\begin{equation*}
P\angle{i} = 
\begin{cases}
    \LCPP\angle{i}  & \text{if $\lsmear(X_b\angle{i} \:\&\: U\angle{i}) =  \lsmear(X_l\angle{i} \:\&\: U\angle{i})$ }, \\
    \LCPP\angle{i} \mid M\angle{i}     & \text{otherwise}.
\end{cases}
\end{equation*}
Recall that $U\angle{i}$ consists of $1$'s at all position in of $1$'s
in all positions to the right of $\lcpp(X_p\angle{i},
Y_p\angle{i})$. Hence, if $\lsmear(X_b\angle{i} \:\&\: U\angle{i}) =
\lsmear(X_l\angle{i} \:\&\: U\angle{i})$, then
$\ldp_{Y_p\angle{i}}(X_p\angle{i})$ is a light part. By
Lemma~\ref{lem:nca} it follows that $P\angle{i}$ is the part label for
$\lnca(X\angle{i}, Y\angle{i})$. To compute $P$, we compare the
sequences $\lsmear(X_b\angle{i} \:\&\: U\angle{i})$ and
$\lsmear(X_l\angle{i} \:\&\: U\angle{i})$, extract fields accordingly
from $M$, and $\mid$ this with $\LCPP$. The remaining sublabels are
constructed by extracting from $X_b$ and $X_l$ using $Z$. We construct
the final $f$-packed sequence $\Lnca(X,Y)$ by zipping the sublabels
together.

The time for this step is $O(1)$.  

\medskip 
The total time for the algorithm is $O(\log f)$. For general packed
sequences, we have the following result.
\begin{lemma}\label{lem:wordlnca}
  For $f$-packed sequences $X$ and $Y$ of length $r$, we can compute
  $\Lnca(X,Y)$ in time $O(\frac{rf}{w}\log f + 1)$.
\end{lemma}
\begin{proof}
  Apply the algorithm from the proof of Lemma~\ref{lem:nstdlnca} using
  the $O(\log f)$ implementation of $\Lnca$ instruction. The time is
  $O(r \log f/s + 1) = O(\frac{rf}{w}\log f + 1)$.
\end{proof}

\subsection{The Algorithm}
We combine the implementation of $\Map$ and $\Lnca$ with
Lemma~\ref{lem:redux} to obtain the following result.
\begin{theorem}\label{thm:word}
  Approximate string matching for strings $P$ and $Q$ of lengths $m$
  and $n$, respectively, with error threshold $k$ can be solved in
  time $O(nk\cdot \frac{\log^{2} m \log w}{w} + n)$ and space $O(m)$.
\end{theorem}
\begin{proof}
  We plug in the results for $\Map$ and $\Lnca$ from
  Lemmas~\ref{lem:wordlnca} and \ref{lem:packedmap} into the reduction
  from Lemma~\ref{lem:redux}. We have $r,u = O(m)$ and $f = O(\log m)$
  and therefore $s = p = O(m)$ and $q = O(\frac{rf}{w}\log r \log w +
  \frac{(r+u)f}{w}\log w + \frac{rf}{w}\log f + 1) = O(\frac{m\log^{2}
    m \log w}{w} + 1)$. Thus, we obtain an algorithm for approximate
  string matching using space $O(m)$ and time $O(\frac{nk}{m} \cdot
  \frac{m\log^{2} m \log w}{w} + n) = O(nk\cdot \frac{\log^{2} m \log
    w}{w} + n)$.
\end{proof}
Combining Theorems~\ref{thm:tabulation} and \ref{thm:word} we have
shown Theorem~\ref{thm:main}.

\section{Acknowledgments}
We would like to thank the anonymous reviewers for many valuable
comments that greatly improved the quality of the paper.

\bibliographystyle{abbrv}
\bibliography{paper}

\end{document}